\documentclass[11pt]{llncs}
\usepackage[margin=1in]{geometry}

\usepackage{amssymb, amsfonts, amsmath}
\usepackage [autostyle, english = american]{csquotes}
\usepackage{tabularx}
\usepackage{graphicx}
\usepackage{url}
\usepackage{bbm}
\usepackage{algorithm}
\usepackage{algorithmic}

\usepackage[T1]{fontenc}
\usepackage[cal=boondox]{mathalfa}

\MakeOuterQuote{"}

\newcommand{\size}[1]{\lvert #1 \rvert}

\newcommand{\figref}[1]{Figure~\ref{fig:#1}}
\newcommand{\tabref}[1]{Table~\ref{tab:#1}}
\renewcommand{\eqref}[1]{Equation~\ref{eq:#1}}

\newcommand{\secref}[1]{Section~\ref{sec:#1}}
\newcommand{\appref}[1]{Appendix~\ref{app:#1}}
\newcommand{\assnref}[1]{Assumption~\ref{assn:#1}}
\newcommand{\assntworef}[2]{Assumptions~\ref{assn:#1} and~\ref{assn:#2}}
\newcommand{\thmref}[1]{Theorem~\ref{thm:#1}}
\newcommand{\thmstworef}[2]{Theorems~\ref{thm:#1} and~\ref{thm:#2}}
\newcommand{\thmsthreeref}[3]{Theorems~\ref{thm:#1}, ~\ref{thm:#2}, and~\ref{thm:#3}}

\newtheorem{assumption}{Assumption}

\begin{document}

\title{Turing Completeness and \emph{Sid Meier's Civilization}}
\titlerunning{Turing Completeness and \emph{Sid Meier's Civilization}}

\author{Adrian de Wynter}
\maketitle

\begin{abstract}\footnote{Preprint. This paper is for scholastic purposes only and the author's affiliations, past and present, should not be considered an endorsement.}

We prove that three strategy video games from the \emph{Sid Meier's Civilization} series: \emph{Sid Meier's Civilization: Beyond Earth}, \emph{Sid Meier's Civilization V}, and \emph{Sid Meier's Civilization VI}, are Turing complete. We achieve this by building three universal Turing machines--one for each game--using only the elements present in the games, and using their internal rules and mechanics as the transition function. The existence of such machines imply that under the assumptions made, the games are undecidable. 
We show constructions of these machines within a running game session, and we provide a sample execution of an algorithm--the three-state Busy Beaver--with one of our machines.
\end{abstract}

\section{Introduction}

The \emph{Sid Meier's Civilization}\footnote{All products, company names, brand names, trademarks, and images are properties of their respective owners. The images are used here under Fair Use for the educational purpose of illustrating mathematical theorems.} series is a collection of turn-based strategy video games, with a focus on building and expanding an empire through technological, social, and diplomatic development. 
The scale and scope of these games involve intricate rules and mechanics, along with multiple victory conditions, all of which makes them interesting from a computational point of view. 

In this paper we focus on leveraging the mechanics inherent to these games to build universal Turing machines. Universal Turing machines, or UTMs, are mathematical abstractions used in the rigorous study of mechanical computation. By definition, a UTM is able to execute any program that can be ran in a computer, and is effectively equivalent to a computer with unbounded memory. 

We introduce explicit constructions of UTMs for three installments of the \emph{Civilization} series: \emph{Sid Meier's Civilization: Beyond Earth} (Civ:BE from now on), \emph{Sid Meier's Civilization V} (Civ:V), and \emph{Sid Meier's Civilization VI} (Civ:VI). Our only strong constraint is that the map size and turn limits must be infinite, which is in line with the requirements around a UTM. 

For the constructions of Civ:BE and Civ:V, we design UTMs that are relatively easy to visualize, although "large"--that is, the transition function has many instructions. We also describe, but not prove, smaller constructions. Finally, we note that an immediate consequence of the existence of an internal game state representing a UTM is undecidability of said game under the stated constraints. We provide constructions of each UTM in their respective games, and conclude by providing an example implementation and execution of the three-state Busy Beaver game \cite{BusyBeaver} in the Turing machine built for Civ:BE.

\section{Background}\label{sec:background}

\subsection{Turing machines and Turing completeness}

A Turing machine can be physically visualized as an automaton that executes instructions as dictated by a function $\delta$, and that involve reading or writing to a tape of infinite length via a head. Following the definition from Hopcroft et al. \cite{HopcroftUllman}, they are defined as a 7-tuple $\langle Q, \Gamma, \Sigma, \delta, b, q_0, F \rangle$, where:
\begin{itemize}
    \item $Q$ is a set of possible \emph{states} that can be taken internally by the Turing machine. Namely, $q_0 \in Q$ is the initial state, and $F \subseteq Q$ is the set of final (or accepting) states. 
    \item $\Gamma$ is the set of tape \emph{symbols} on which the Turing machine can perform read/write operations. In particular, $b \in \Gamma$ is the blank symbol, and $\Sigma \subseteq \Gamma$, $b \not\in \Sigma$ is the set of symbols which may be initially placed in the tape.
    \item $\delta : (Q \setminus F) \times \Gamma \nrightarrow Q \times \Gamma \times \{L, R\}$ is the \emph{transition function}. It is a partial function that takes in all states from the machine, along with the current read from the tape; it determines the next state to be taken by the Turing machine, as well as whether to move the head left ($L$) or right ($R$) along the tape. Note that not moving the tape can be trivially encoded in this definition.
\end{itemize}

Turing machines are a mathematical model, and are limited in practice due to the infinite length requirement on the tape. On the other hand, the expressive power behind the theory of Turing machines allows for the rigorous study of computational processes. 

Concretely, since Turing machines are finite objects, it is possible to encode a Turing machine $M_1$ and use it as an input to a separate machine $M_2$, so that $M_2$ simulates $M_1$. A $(m, n)$-\emph{universal Turing machine}, or $(m, n)$-UTM, is a Turing machine that is able to simulate any other Turing machine, and such that $\size{Q} = m, \size{\Gamma} = n$. The Church-Turing Thesis states that such a construct models precisely the act of computation--concretely, the set of functions that can be computed by a UTM is precisely the set of computable functions \cite{RogersComputability}. 

A collection of rules and objects under which a simulation of a UTM is possible is said to be \emph{Turing complete}. Turing complete models of computation are varied, from mathematical models such as the $\lambda$-calculus and the theory of $\mu$-recursive functions; to other rule-based systems like Conway's Game of Life \cite{ConwaysGOL} and the card game \emph{Magic: The Gathering} \cite{Churchill}; and even video games and other software such as \emph{Minesweeper} \cite{Minesweeper} and Microsoft PowerPoint \cite{PowerPoint}.

\subsection{General Mechanics of \emph{Sid Meier's: Civilization}}

All of the \emph{Sid Meier's: Civilization} installments are turn-based, alternating, multiplayer games. The players are in charge of an empire,\footnote{In the case of Civ:BE, an extrasolar colony.} and good management of resources and relationships with other players are key to achieve one of the multiple possible victory conditions allowable. These victories (e.g., Diplomatic, Militaristic, Scientific) and their specific mechanics are of interest when analyzing the computational complexity of the game, but not when designing a Turing machine. 
In this paper we limit ourselves to describing the rules governing the actions of units on the map, and which are common to the games that we will be analyzing. Since the naming conventions for the components of the games vary across them--even though these components may serve the same purpose--we also introduce a unified notation. 

The games that we will be discussing (Civ:BE, Civ:V, and Civ:VI) are played on a finite map divided in hexagonal tiles, or \emph{hexes}. Hexes are composed of different geographical types (e.g., oceans, deserts, plains), and may contain certain features (e.g., forests, oil) dependent on the initial conditions of the game. 

All tiles have an inherent \emph{resource} yield, which is used by the player to maintain the units and buildings that they own, and to acquire more advanced \emph{technologies} and \emph{social policies}.
Some examples of these resources are Food, Production, Culture, and Science. 
One or more player-controlled units known as \emph{Workers} are usually tasked to build \emph{improvements} on a tile, which alters their yields. 
As an example, technologies are acquired with Science, and Workers can build improvements on tiles (e.g. an Academy) to increase their Science yield. 

With some notable exceptions, described in the following section, Workers may only build improvements in tiles owned by a \emph{City}. A City can be seen as a one-tile special improvement, and has three main duties: to produce ("train") units, to control and expand the boundaries of the owner's frontiers, and to produces \emph{Citizens} to work the tiles and their improvements. Citizens may only work the tiles owned by the City, but players are able to micromanage the placement of Citizens to alter their resource yield. 

It is important to note that on every turn, the player executes all or some of the actions available to them, such as building improvements, researching technologies and social policies, and focuses on the overall management of the empire. 
On the other hand, the specific types and mechanics around Cities, resources, features, available improvements, technologies and social policies vary from game to game. They will be discussed in detail in the next section, as they will be the key for building our UTMs. All of the UTMs are built inside a running game session, and leverage the concepts described here.

\section{Civ:BE, Civ:V, and Civ:VI are Turing Complete}\label{sec:turinguniv}

For this section, the following two assumptions hold:

\begin{assumption}\label{assn:assumption1}
The number of turns in Civ:BE, Civ:V, and Civ:VI is infinite, and there is only one player in the game.
\end{assumption}
\begin{assumption}\label{assn:assumption2}
The map extends infinitely in either direction.
\end{assumption}

The rationale behind \assnref{assumption1} is related to the fact that another player may interfere with the computation by--say--attacking the Workers, or reaching any of the victory conditions before the Turing machine terminates. Likewise, \assnref{assumption2} is due to the fact that we intend to utilize the map as the tape for the UTM. We will discuss the feasibility of these assumptions in \secref{conclusion}. For now, it suffices to note that \assnref{assumption1} is easily achievable within the game itself; and \assnref{assumption2} is a requirement present in all UTM constructions, and inherent to the definition of a Turing machine.

\subsection{A $(10, 3)$-UTM in Civ:BE}

For the construction of this UTM, we rely on two Workers carrying out instructions on separate, non-overlapping sections of the map: one keeps track of states, and another acts as the head. The latter operates over the tape, which we consider to be the rest of the map. We do not require the player to own any of the tiles on the tape. However, the state hexes are finite and must be owned by the player. Building any improvement in Civ:BE takes a certain amount of turns, but the Workers may build and repair any number of them. 

\subsubsection{The Tape:}

The tape is a continuous strip over the entire map, save for the state tiles, and without loss of generality, we assume that it is comprised of flat, traversable, hexes. Flatness is needed because irregular terrain requires more movement points to traverse. 
The set of symbols $\Gamma$ is based off specific improvements that can be added and removed indefinitely by the tape Worker anywhere on the map--namely, Roads. The tape Worker adds and removes Roads according to the transition function, and a fast-moving unit (in this case, a Rover) pillages it. Note that for this setup to work, the Worker and the Rover will must move at the same time. Then, $\Gamma = \{\text{No Improvement}, \text{Road}, \text{Pillaged Road}\}$.

\subsubsection{The States:}

We assume that the player has at least nine flat, desert, unimproved tiles which will act as the states. Desert is required because we  map states to the resource yields provided by a specific improvement--the \emph{Terrascape}. Income from other tile types may interfere with state tracking. 
Moreover, we assume the player has unlocked the \emph{civic} (social policy) \emph{Ecoscaping}, which provides $+1$ Food, Production, and Culture for every Terrascape built, along with the relevant technology (\emph{Terraforming}) required to build them. The Worker in charge of the states builds and remove Terrascapes on the state hexes, and the normalized change in Culture income serves as the state value $q_i \in Q$, $Q = \{ 0, \dots, 9 \}$. 

See \figref{civbe} for a sample construction. 

\begin{figure}
\centering
\includegraphics[width=\columnwidth]{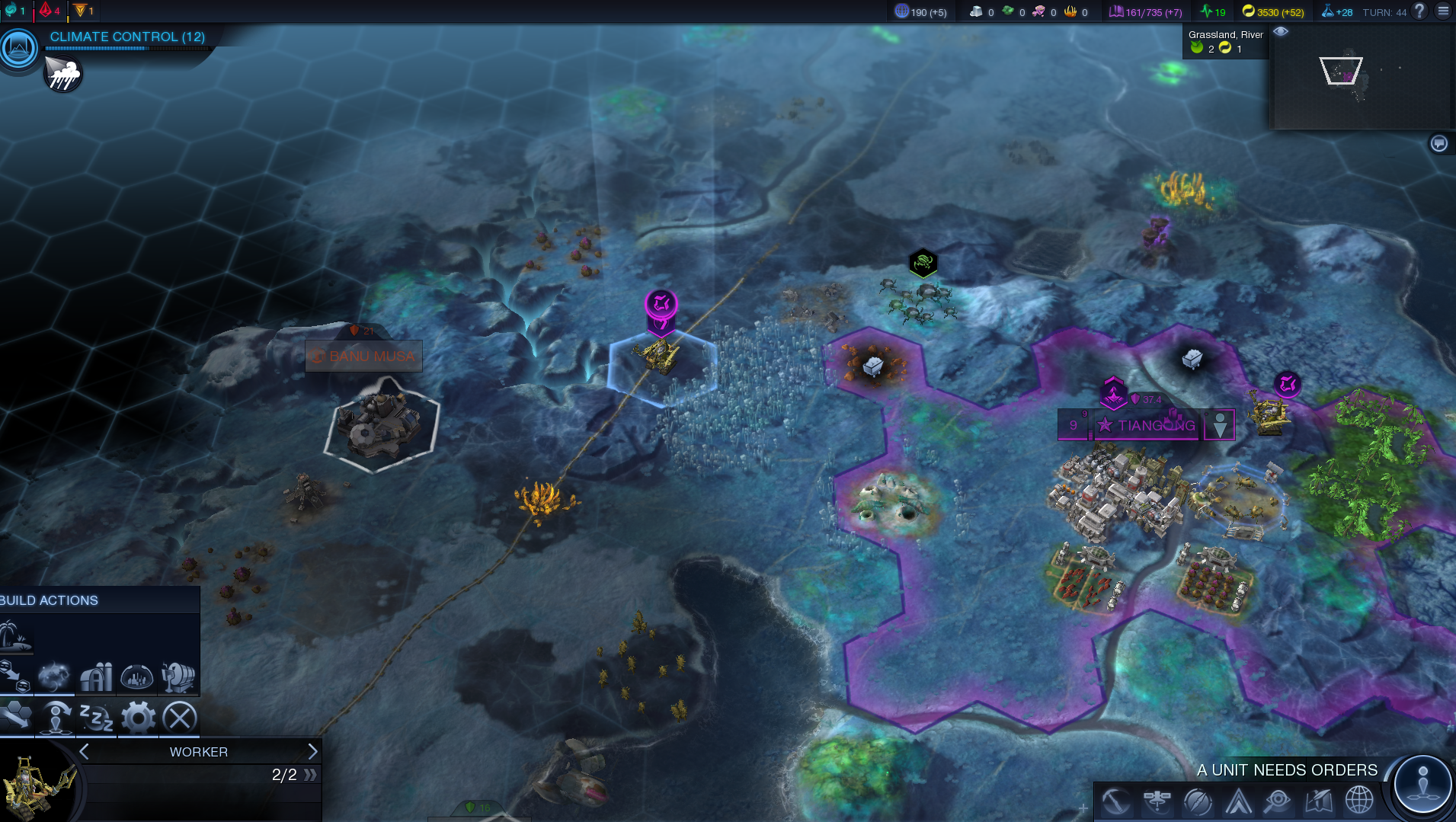}
\caption{A $(10, 3)$-UTM built within Civ:BE. The tape is located to the left of the image, and runs from the top middle of the screen to the bottom left. It currently has Roads (symbols in $\Gamma$) on every hex, except one pillaged Road. The head (the tape Worker and the Rover) are positioned in the middle of the screen, reading the symbol "Road"--the read is implicit, but can also be seen from the "Build Actions" menu, which does not allow the player to build a Road on an existing Road. The UTM is in state $q_2$, since the total Culture yield in this turn is $C_t = 7$ (the purple number, top right on the image). Each Terrascape (lush green tiles on the right) contributes $3$ Culture, and the base Culture yield is $C_* = 1$, so $q_2 = (C_t - C_*)/3 = 2$. 
For this particular map, desert tiles appear as solid ice.}
\label{fig:civbe}
\end{figure}

\begin{theorem}\label{thm:civbeutm}
Suppose \assntworef{assumption1}{assumption2} hold. 
Also suppose that there is a section of the map disjoint from the rest, owned by the player, and comprised of at least nine desert, unimproved tiles; and that the player is able to build and maintain at least nine worked Terrascapes, has a constant base per-turn Culture yield $C_*$, and has the Ecoscaping civic. 

Then the following is a $(10, 3)$-UTM:

\begin{align}
Q &= \{ 0, \dots, 9 \}, \\
\Gamma &= \{ \text{No Improvement}, \text{Road}, \text{Pillaged Road} \},
\end{align}

where $q_i \in Q$ is the difference in the normalized Culture yields in turn $t$, $q_i = C_t - C_*$, and the transition function $\delta$ is as described in \tabref{civbe}.

\end{theorem}
\begin{proof}
A full proof can be found in \appref{civbeapp}. It consists of two parts: building a bijection between the transition function and an existing $(10, 3)$-UTM, and then showing that the time delay between our construction and the aforementioned UTM is bounded by a constant.
\end{proof}

\begin{remark}
This is not the only possible UTM that can be built within Civ:BE. Unlocking the ability of a worker to add and remove Miasma expands the set of symbols from the tape yields a $(7, 4)$-UTM with a smaller set of instructions. 
\end{remark}

There are a few things to highlight from this construction, and that will apply to the rest of the UTMs presented in this paper. 
First, while the base constant yield is hard to achieve in-game, it is relatively simple to account for: simply execute every move at the beginning of the turn, record $C_*$, and then execute the transition function for the UTM. Even better, counting the number of Terrascapes on a specific tile would work just as well.

Second, remark that this UTM, along with the others presented in the paper, may be fully automated by using the API provided by the publisher. Manual simulation is, although tiresome, also possible, as displayed in \secref{busybeaver}. 

Finally, it is well-known in the field of theoretical computer science that it is possible to have a UTM with no states, by simply increasing the number of heads and using a different section of the tape as scratch work. Such a construction will be equivalent to ours. 

\subsection{A $(10, 3)$-UTM in Civ:V}

We follow the same approach as in \thmref{civbeutm}: one Worker builds improvements (Roads and Railroads) on tiles to simulate the symbols on a tape, and another Worker improves and unimproves tiles to encode the internal state of the machine based on the relative yield of a specific resource. Just as in Civ:BE, building any improvement in Civ:V takes a certain amount of turns, but the Workers may build any number of them. 

\subsubsection{The Tape:}

As in \thmref{civbeutm}, the tape is a continuous strip over the entire map, and it is comprised of flat, traversable, hexes. 
The set of symbols $\Gamma$ is the two improvements that can be built by the tape Worker anywhere on the map: Roads and Railroads. Building a Railroad requires that the player has unlocked the \emph{Engineering} technology. Then $\Gamma = \{ \text{No Improvement}, \text{Road}, \text{Railroad} \}$.

\subsubsection{The States:}

In Civ:V it is no longer possible to remove improvements outside of Roads and Railroads, so we must rely on these improvements for state tracking. 
The Worker in charge of the states builds and removes Railroads to encode the state. This is carried out in a reserved section of the map (e.g., inside a City), and the total number of Railroads serves as the state value $q_i \in Q$, $Q = \{ 0, \dots, 9 \}$.

\begin{theorem}\label{thm:civvutm}
Suppose \assntworef{assumption1}{assumption2} hold. 
Also suppose that there is a section of the map disjoint from the rest, owned by the player, and comprised of at least nine tiles.
Assume that the player has 
the Engineering technology. 

Then the following is a $(10, 3)$-UTM:

\begin{align}
Q &= \{ 0, \dots, 9 \}, \\
\Gamma &= \{ \text{No Improvement}, \text{Road}, \text{Railroad} \},
\end{align}

where $q_i \in Q$ is the total number of Railroads in the nine tiles, 
and the transition function $\delta$ is as described in \tabref{civv}.

\end{theorem}
\begin{proof}
In \appref{civvapp}; it is similar to the proof of \thmref{civbeutm}. 
\end{proof}

\begin{figure}
\centering
\includegraphics[width=\columnwidth]{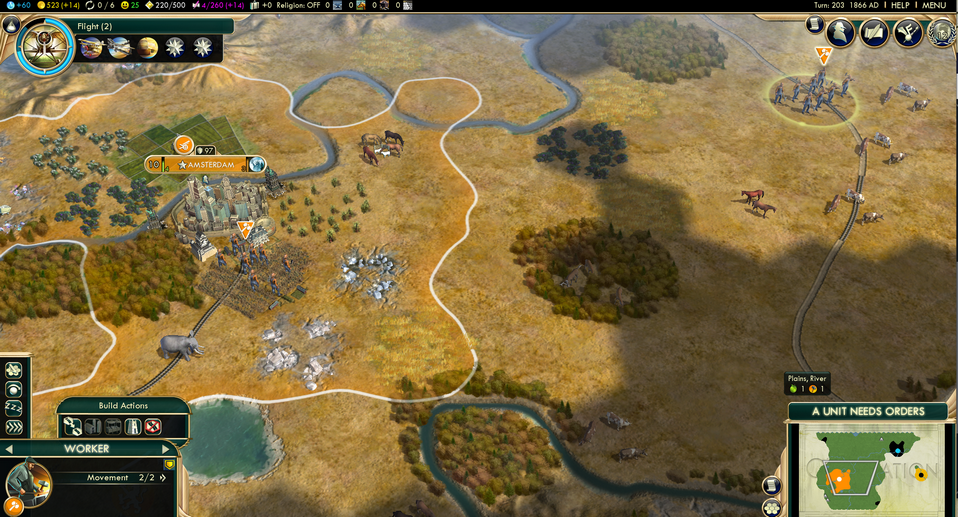}
\caption{A $(10, 3)$-UTM built within Civ:V. It is in the state $q_2$ since the City (left) has two Railroads built: one on the tile with the elephant and another on the plantation. In this case, the area reserved for state tracking is every tile owned by the City. The head (the rightmost Worker, near the top) is positioned over a Railroad hex with two Roads adjacent to it.}
\label{fig:civv}
\end{figure}

\begin{remark}
Just as in \thmref{civbeutm}, it is possible to use a fast-moving pillaging unit (e.g., a Knight) and the ability of Workers to build Forts, to expand the symbol set of the machine to a $(5, 5)$-UTM. This machine would have an even smaller set of instructions. 
\end{remark}

\subsection{A $(48, 2)$-UTM in Civ:VI}

In Civ:VI, Worker units are unable to build improvements indefinitely. It is still possible to build an UTM in Civ:VI if we use the Cities as the head \emph{and} the tape, and leave the Worker to encode only the position of the head. The Citizens are micromanaged to "write" symbols to the tape, by putting them to work specific hexes. As before, we maintain a separate set of tiles to act as the states.

In CiV:VI, there is no upper limit on the number of Cities founded by the player, as long as they are all placed at least $4$ hexes away from other Cities. 
Under \assnref{assumption1}, this allows us to found an infinite number of Cities via a unit called a \emph{Settler}. These units are trained by the City, and once the Settler has been created, the City loses a citizen. Provided there is enough Food, a new Citizen is born in the City after a certain number of turns. No Settlers can be trained if the City has only one citizen. 

Founding a City consumes the Settler, and the City takes its place. The tiles owned at this time are set to be the six closest hexes,\footnote{This action technically builds a City Center, but we will forgo this distinction. Players who play as Russia will found Cities with five extra hexes.} and two Citizens spawn: one works the tile with the highest-possible yield of combined resources, and another--which cannot be retasked--the City itself. All of this occurs on the same turn. 

We will utilize the available tiles and the Citizens to build our $(48, 2)$-UTM. A sample construction can be seen in \figref{civvi}. Although our construction is more complicated than the previous UTMs, we note that all of this is achievable with the built-in mechanics and rules of the game.

\subsubsection{The Tape:}
Assume, for simplicity, that the map is comprised of flat hexes, all of which are of the floodplains category. A single, uninterrupted, grassland strip is the only exception. 
The tape is then the grassland tiles owned by the player. Cities are collocated three hexes away from one another over the grassland strip. 
The symbols on the tape ($\Gamma = \{$"is being worked", "is not being worked"$\}$) are the number of Citizens working on a given grassland tile. 

Our design is unconventional since the tape is not infinite by design. Instead, we consider a transition to be concluded if the City has at least three Citizens. We cap the growth of any new City to four Citizens: one that works in the City (and cannot be moved), two to work the tape tiles, and one that will act as a Settler if need be. If a Settler has been created, then the growth of the City is capped at three. Note that this doubles as a marker for the end of the tape. 

Every time the transition function issues a "move right" command, the controller first checks the population size of the City, and the position of the Worker unit. There are four possible cases:

\begin{itemize}
\item The City has three Citizens, and the Worker is on the rightmost (resp. leftmost) hex. Then the Worker moves right (left), to the next City, and is placed on the leftmost (rightmost) hex.
\item The City has four Citizens. This is the end of the tape and a Settler must be trained to expand the tape. After it is spawned, the City is capped to three, and the Settler is issued the command to move right (resp. left) four hexes, build a City, and grow it to four population. The Worker is issued the command to move right (left) three hexes, and wait until the City is grown.
\end{itemize}

The "move left" command is symmetric to the above.

\subsubsection{The States:}
Just as in \thmstworef{civbeutm}{civvutm}, we rely on a set of hexes disjoint to the tape for our state tracking, and we utilize the relative yield of a specific resource (Faith) to encode it. We assume the player has built $23$ Monasteries without any adjacent districts. Monasteries, when worked under these conditions, yield $+2$ Faith per turn, and are only buildable if the player has ever allied with the Armagh City-State. We also require the player to have built $23$ Farms without any adjacent districts.\footnote{As in the other constructions, we impose specific requirements for our proofs and their generalizability. In practice, any resource yield (e.g., Culture) or improvement (e.g., Ethiopia's Rock-Hewn Church; Egypt's Sphinx; India's Stepwell) would work.} 
Transitioning from one state to the other requires the player to re-assign as many workers as possible to adjust the difference in Faith yields from the beginning of time, which in turn realizes (part of) the state $q_i \in Q$. 

\begin{theorem}\label{thm:civviutm}
Suppose \assntworef{assumption1}{assumption2} hold. Also suppose that the map contains an infinite strip of grassland tiles, with floodplains above and below, and that the player owns one City placed in the grassland strip. 
Also assume that there is a set of tiles disjoint to the strip, owned by the player, and with $23$ Monasteries and $23$ Farms without any adjacent districts. Let $F_*$ be the constant base Faith yield per turn for this player. 

Then, if there are no in-game natural disasters occuring during the time of the execution, the following Turing machine is a $(48, 2)$-UTM:

\begin{align}
Q &= \{ 0n, \dots, 23n \} \cup \{ 0b, \dots, 23b \} , \\
\Gamma &= \{ \text{Is Being Worked}, \text{Is Not Being Worked} \},
\end{align}

where $q_ix \in Q$ is the difference in the Faith yields at turn $t$, $q_i = F_t - F_*$; $x\in \{n, b\}$ indicates whether it is needed to build a new Settler, and the transition function $\delta$ is as described in \tabref{civvi}.

\end{theorem}
\begin{proof}
In \appref{civviapp}.
\end{proof}

\begin{figure}
\centering
\includegraphics[width=\columnwidth]{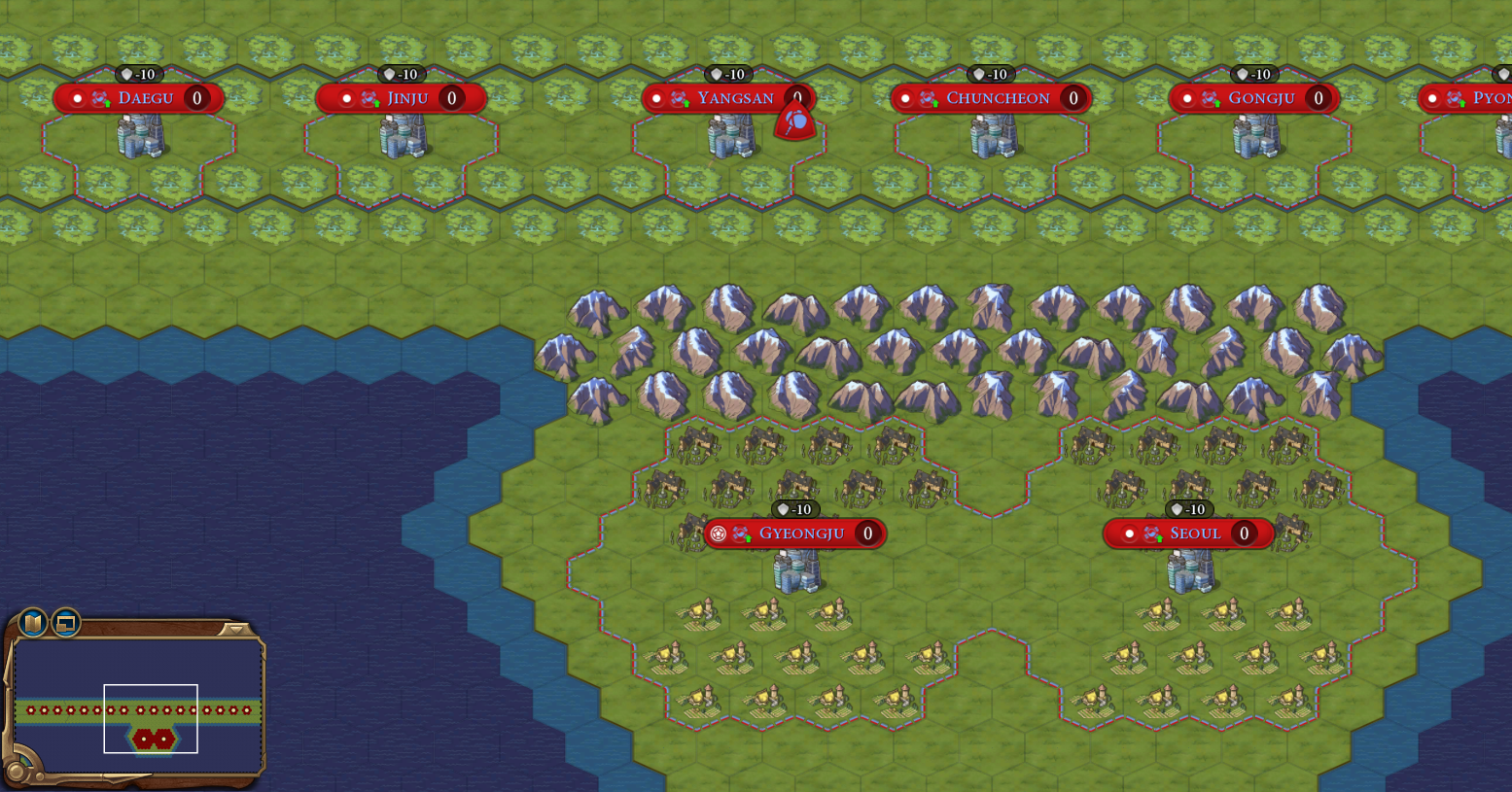}
\caption{Diagram of a $(48, 2)$-UTM built with Civ:VI components and rules. The bottom two Cities act as the state of the machine, and each of the Cities on the top strip contain two cells from the tape. The head position is tracked by the Worker. Refer to \figref{igcivvi} for a working in-game construction.}
\label{fig:civvi}
\end{figure}

In \appref{civviappalt} we discuss an alternative implementation involving two players at war with each other, and which requires fewer cities being built.

\begin{remark}
This construction works with every civilization and update released up to and including the April 2021 Update. It is, however, much more brittle than the other two UTMs introduced in this paper.
\end{remark}

\begin{corollary}
Civ:BE, Civ:V, and Civ:VI are undecidable under \assntworef{assumption1}{assumption2}.
\end{corollary}
\begin{proof}
Immediate from the proofs of \thmsthreeref{civbeutm}{civvutm}{civviutm}. The existence of an in-game state that simulates a UTM implies undecidability.
\end{proof}

\section{Example: The Busy Beaver in Civ:BE}\label{sec:busybeaver}

\begin{table}\label{tab:bbtm}
\begin{center}
\caption{Transition function for the three-state Turing machine corresponding to BB-$3$, and the equivalent Civ:BE construction. For the Turing machine, the notation is read as $(state, \;read;\; write,\;move,\;set\;state)$. For the Civ:BE machine, $\Delta$ corresponds to the normalized change in Culture per turn.}
\def\arraystretch{1.05}
\begin{tabular}{|l|l||c|}
\hline
Civ:BE state ($\Delta;\; tape\; read$) \;&\; Command to Workers $(tape;\; state)$ & BB-$3$ TM 
 \\ \hline
$\;0$; No Improvement & \;Build a Road and move $R$; Build a Terrascape    \;&  $q_00;1 R q_1$ \\
$\;0$; Road           & \;Build a Road and move $L$; Build $2$ Terrascapes \;&  $q_01;1 L q_2$ \\
\hline
$\;1$; No Improvement & \;Build a Road and move $L$; Remove a Terrascape \;&  $q_10;1 L q_0$ \\
$\;1$; Road           & \;Build a Road and move $R$; No build            \;&  $q_11;1 R q_1$ \\
\hline
$\;2$; No Improvement & \;Build a Road and move $L$; Remove a Terrascape \;&  $q_20;1 L q_1$ \\
$\;2$; Road           & \;HALT                                           \;&  $\;q_21;1 R$ $HALT$ \\
\hline
\end{tabular}
\end{center}
\end{table}

The Busy Beaver game, as described by Rad\'o \cite{BusyBeaver}, is a game in theoretical computer science. 
It challenges the "player" (a Turing machine or a person) to find a halting Turing machine with a specified number of states $n$, such that it writes the most number of $1$s in the tape out of all possible $n$-state Turing machines. The winner is referred to as the $n^{\text{th}}$ Busy Beaver, or BB-$n$. Note that all the Turing machines playing the Busy Beaver game have, by definition, only two symbols, $\Gamma = \{0, 1\}$.

Finding whether a given Turing machine is a Busy Beaver is undecidable, but small Busy Beavers are often used to demonstrate how a Turing machine works. In particular, BB-$3$ is a three-state ($Q = \{q_0, q_1, q_2\}$), two-symbol ($\Gamma = \{0, 1\}$) Turing machine with a transition function as described in \tabref{bbtm}. 

Given the Turing completeness of Civ:BE, we can build BB-$3$ within the game itself. An equivalent construction with our Civ:BE $(10, 3)$-UTM has the states $Q = \{0, 1, 2\}$, symbols $\Gamma = \{\text{No Improvement}, \text{Road}\}$, and with a transition function as described in \tabref{bbtm}. A sample execution of BB-$3$ within Civ:BE can be seen in \figref{civbebb}.\footnote{An animated version of this execution can be found in \url{https://adewynter.github.io/notes/img/BB3.gif}.}

\begin{figure}
\centering
\includegraphics[width=\columnwidth]{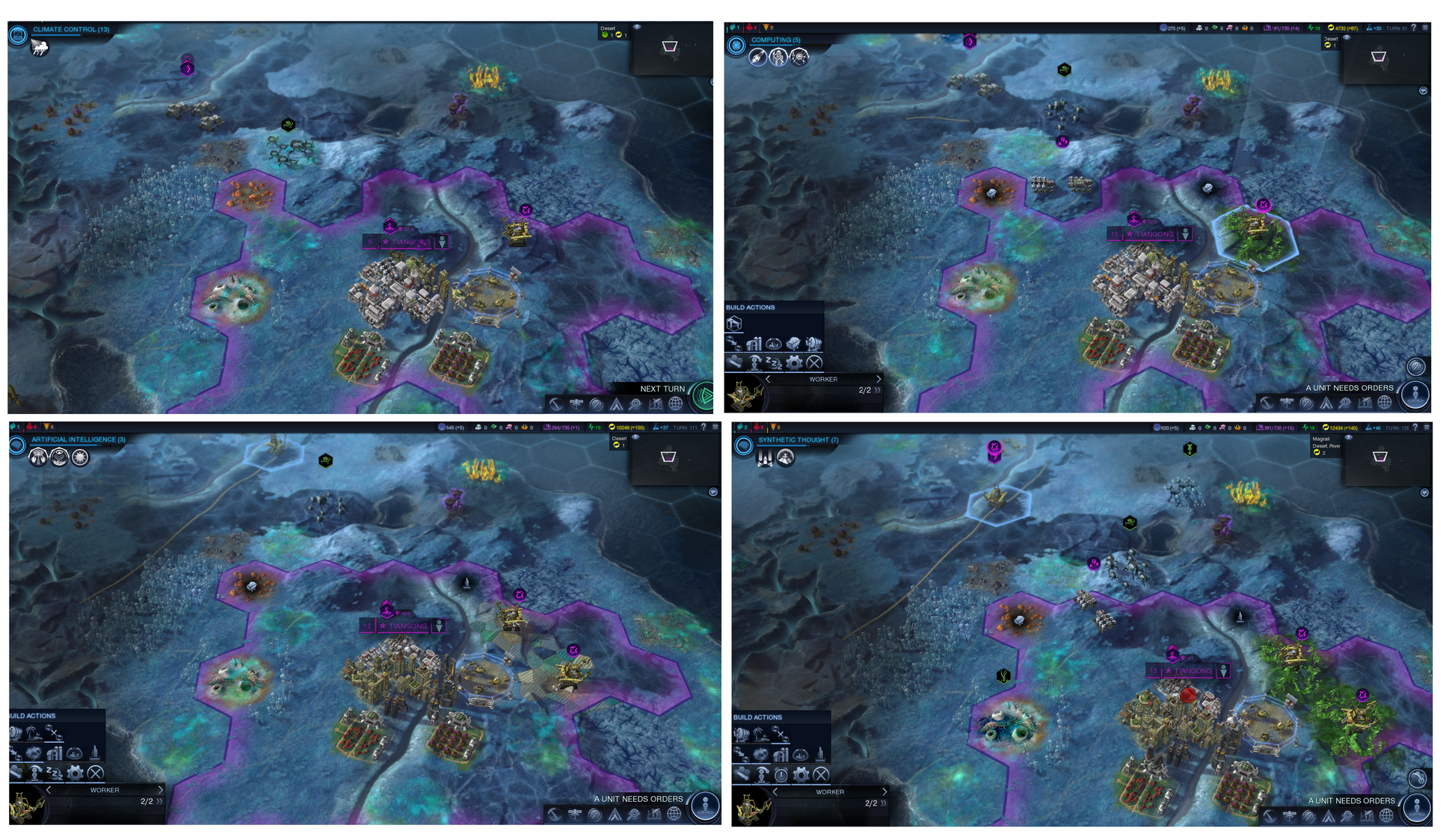}
\caption{Sample execution of BB-$3$ with a Civ:BE Turing machine. The machine executes $11$ instructions $t_1, \dots, t_{11}$ before halting. \emph{Top left}: state corresponding to $t_1 =  q_00;1Rq_1$. \emph{Top right}: state for the machine at $t_2 =  q_10;1Lq_0$. \emph{Bottom left}: state for the machine at $t_{10} = q_01;1Lq_2$. \emph{Bottom right}: state for the machine at $t_{11} =  q_21;HALT$.}
\label{fig:civbebb}
\end{figure}

\section{Discussion}\label{sec:conclusion}

We introduced UTMs for Civ:BE, Civ:V, and Civ:VI, and showed their ability to execute an arbitrary algorithm by running BB-$3$ on our Civ:BE UTM. We also proved that, as an immediate consequence of the existence of a UTM within their internal state, these games are undecidable under the unbounded turn limit and map size assumptions. 

Turing completeness in a system goes beyond a mere curiosity, however: the work here showed that there might exist states of the game where it is undecidable to determine whether there exists a sequence of moves that yields a specific state; and, more importantly, that it is possible to write and execute programs in-game. 

That being said, any realistic construction of a Turing machine is physically limited by the hardware. It follows that if these assumptions were relaxed, Civ:BE, Civ:V, and Civ:VI are equivalent to a linear-bounded automaton. Even then, the games are characterized by long sessions, complex rules, multiple victory conditions, varying scales, and imperfect information, all of which leave room for further analysis from a computational complexity perspective. While other video games have been studied, such as classic Nintendo games \cite{NintendoNP} and traditional board games (both of which led to the development of powerful theoretical tools \cite{DemaineGames}), most 4X video games have not analyzed from this angle. 

A good understanding of the computational complexity associated with these games is important for the development of more efficient and engaging AI systems; and their inherent long-term simulative nature at the macro and micro-management level makes them correlate into contemporary problems in computer science, such as multiple-scale forecasting and online learning. 

\bibliography{biblio}

\begin{thebibliography}{10}
\providecommand{\url}[1]{{#1}}
\providecommand{\urlprefix}{URL }
\expandafter\ifx\csname urlstyle\endcsname\relax
  \providecommand{\doi}[1]{DOI~\discretionary{}{}{}#1}\else
  \providecommand{\doi}{DOI~\discretionary{}{}{}\begingroup
  \urlstyle{rm}\Url}\fi

\bibitem{NintendoNP}
Aloupis, G., Demaine, E.D., Guo, A., Viglietta, G.: Classic {N}intendo games
  are (computationally) hard.
\newblock Theoretical Computer Science \textbf{586}, 135--160 (2015).
\newblock \doi{https://doi.org/10.1016/j.tcs.2015.02.037}.
\newblock
  \urlprefix\url{https://www.sciencedirect.com/science/article/pii/S0304397515001735}.
\newblock Fun with Algorithms

\bibitem{ConwaysGOL}
Berlekamp, E.R., Conway, J.H., Guy, R.K.: Winning Ways for your Mathematical
  Plays.
\newblock A K Peters Ltd., Wellesley, Massachusetts (2004)

\bibitem{Churchill}
Churchill, A., Biderman, S., Herrick, A.: Magic: {T}he {G}athering is {T}uring
  complete.
\newblock CoRR \textbf{abs/1904.09828} (2019).
\newblock \urlprefix\url{http://arxiv.org/abs/1904.09828}

\bibitem{DemaineGames}
Hearn, R.A., Demaine, E.D.: Games, Puzzles, and Computation.
\newblock A K Peters/CRC Press (2009)

\bibitem{HopcroftUllman}
Hopcroft, J.E., Motwani, R., Ullman, J.D.: Introduction to Automata Theory,
  Languages, and Computation.
\newblock Pearson (2013)

\bibitem{Minesweeper}
Kaye, R.: Infinite versions of {M}inesweeper are {T}uring complete.
\newblock \urlprefix\url{http://web.mat.bham.ac.uk/R.W.Kaye/minesw/infmsw.pdf}.
\newblock Accessed 4/20/2021

\bibitem{BusyBeaver}
Rad\'o, T.: On non-computable functions.
\newblock The Bell System Technical Journal \textbf{41}(3), 877--884 (1962).
\newblock \doi{10.1002/j.1538-7305.1962.tb00480.x}

\bibitem{RogersComputability}
Rogers Jr., H.: The Theory of Recursive Functions and Effective Computability.
\newblock MIT Press, Cambridge, MA (1987)

\bibitem{Rogozhin}
Rogozhin, Y.: Small universal {T}uring machines.
\newblock Theoretical Computer Science \textbf{168}(2), 215--240 (1996).
\newblock \doi{10.1016/S0304-3975(96)00077-1}

\bibitem{PowerPoint}
Wildenhain, T.: On the {T}uring {C}ompleteness of {MS} {P}ower{P}oint.
\newblock
  \urlprefix\url{https://www.andrew.cmu.edu/user/twildenh/PowerPointTM/Paper.pdf}.
\newblock Accessed 4/20/2021

\end{thebibliography}

\section{Appendices}

\subsection{Civ:BE UTM Proof}\label{app:civbeapp}

In this section we prove \thmref{civbeutm}. 
\tabref{civbe} shows a bijective function between the Civ:BE UTM transition function and the $(10, 3)$-UTM program from Rogozhin \cite{Rogozhin}. 
The execution of every command takes at most the time taken to build (or remove) five Terrascapes and one Road. Let $T$ and $M$ the number of turns needed to build a Terrascape and a Road, respectively. Repairing and removing improvements takes only one turn. Then the maximum overhead, for any instruction, will be at most $5T + M$ turns. It follows that this construction is less efficient than a UTM by a constant factor.

\begin{table}\label{tab:civbe}
\begin{center}
\caption{Bijection between the UTM from \thmref{civbeutm} and Rogozhin's $(10, 3)$-UTM. For the Turing machine, the notation is read as $(state, \;read;\; write,\;move,\;set\;state)$. For the Civ:BE machine, $\Delta$ corresponds to the change in Culture per turn relative to the player's base yield.}
\def\arraystretch{1.05}
\begin{tabular}{|l|l||c|}
\hline
Civ:BE state ($\Delta;\;tape\;read$) \;&\; Command to Workers ($tape;\;state$) & $(10, 3)$-UTM  
 \\ \hline
$\;0$; No Improvement &\; Build a Road and move $R$;       No build              \;&  $q_0 0;1 R q_0$ \\
$\;0$; Road           &\; Remove Improvement and move $L$; Build a Terrascape    \;&  $q_0 1;0 L q_1$ \\
$\;0$; Pillaged Road  &\; No Improvement and move $R$;     Build $3$ Terrascapes \;&  $q_0 b;b R q_3$ \\

\hline

$\;1$; No Improvement &\; No Improvement and move $L$;     Build a Terrascape \;&  $q_1 0;0 L q_2$ \\
$\;1$; Road           &\; Remove Improvement and move $L$; No build           \;&  $q_1 1;0 L q_1$ \\
$\;1$; Pillaged Road  &\; No Improvement and move $L$;     No build           \;&  $q_1 b;b L q_1$ \\

\hline

$\;2$; No Improvement &\; No Improvement and move $L$;  Remove a Terrascape   \;& $q_2 0;0 L q_1$\\
$\;2$; Road           &\; Pillage the Road and move $L$; Build $3$ Terrascapes\;& $q_2 1;b L q_5$\\
$\;2$; Pillaged Road  &\; No Improvement and move $R$;  Remove $2$ Terrascapes\;& $q_2 b;b R q_0$\\

\hline

$\;3$; No Improvement &\; Build a Road and move $R$;   Remove $3$ Terrascapes\;& $q_3 0;1 R q_0$\\
$\;3$; Road           &\; No Improvement and move $R$; Build a Terrascape    \;& $q_3 1;1 R q_4$\\
$\;3$; Pillaged Road  &\; Repair the Road and move $L$;   No build           \;& $q_3 b;1 L q_3$\\

\hline

$\;4$; No Improvement &\; Build a Road, Pillage it, and move $L$; Remove $2$ Terrascapes \;& $q_4 0;b L q_2$ \\
$\;4$; Road           &\; No Improvement and move $R$;  No build                         \;& $q_4 1;1 R q_4$ \\
$\;4$; Pillaged Road  &\; No Improvement and move $R$;  No build                         \;& $q_4 b;b R q_4$ \\

\hline

$\;5$; No Improvement &\; Build a Road and move $L$;   Build a Terrascape \;& $q_5 0;1 L q_6$ \\
$\;5$; Road           &\; No Improvement and move $L$; No build           \;& $q_5 1;1 L q_5$ \\
$\;5$; Pillaged Road  &\; No Improvement and move $L$; No build           \;& $q_5 b;b L q_5$ \\

\hline

$\;6$; No Improvement &\; No Improvement and move $R$; Build a Terrascape   \;&  $q_6 0;0 R q_7$ \\
$\;6$; Road           &\; HALT                                              \;&  $\;q_6 1;\;  HALT $ \\
$\;6$; Pillaged Road  &\; No Improvement and move $L$; Build $2$ Terrascapes\;&  $q_6 b;b L q_8$ \\

\hline

$\;7$; No Improvement &\; Build a Road and move $L$;   Remove $2$ Terrascapes \;& $q_7 0;1 L q_5$\\
$\;7$; Road           &\; No Improvement and move $R$; No build               \;& $q_7 1;1 R q_7$\\
$\;7$; Pillaged Road  &\; No Improvement and move $R$; No build               \;& $q_7 b;b R q_7$\\

\hline

$\;8$; No Improvement &\; Build a Road and move $L$;       Build a Terrascape     \;& $q_8 0;1 L q_9$\\
$\;8$; Road           &\; Remove Improvement and move $R$; Build a Terrascape     \;& $q_8 1;0 R q_9$\\
$\;8$; Pillaged Road  &\; Remove Improvement and move $L$; Remove $5$ Terrascapes \;& $q_8 b;0 L q_3$\\

\hline

$\;9$; No Improvement &\; No Improvement and move $R$;     Remove $5$ Terrascapes \;& $q_9 0;0 R q_4$ \\
$\;9$; Road           &\; Remove Improvement and move $R$; No build               \;& $q_9 1;0 R q_9$ \\
$\;9$; Pillaged Road  &\; No Improvement and move $R$;     Remove a Terrascape    \;& $q_9 b;b R q_8$ \\

\hline

\end{tabular}
\end{center}
\end{table}

\subsection{Civ:V UTM Proof}\label{app:civvapp}

In this section we prove \thmref{civvutm}. 
\tabref{civv} shows a bijective function between a Civ:V UTM transition function and the $(10, 3)$-UTM program from Rogozhin \cite{Rogozhin}. 
The execution of every command will take at most the time taken to build (or remove) five Railroads, in addition to removing a Railroad and building a Road. This is because the mechanics of Civ:V do not allow the Workers to build a Road on a Railroad tile, so the Railroad must be removed first. This last action takes one turn. 
Let $B_{rr} > 1$ be the cost of building a Railroad. The Worker in charge of states will have to build at most $3$ Railroads; and the Worker in charge of the tape will have to at most remove a Road and build a Railroad instead, at a cost of $1 + B_{rr}$ turns. 
Then the maximum overhead, for any instruction, is at most $4B_{rr} + 1$ turns. It follows that this construction, like the UTM presented in \thmref{civbeutm}, is less efficient than a UTM by only a constant factor.

\begin{table}\label{tab:civv}
\begin{center}
\caption{Bijection between the UTM from \thmref{civvutm} and Rogozhin's $(10, 3)$-UTM. For the Turing machine, the notation is read as $(state, \;read;\; write,\;move,\;set\;state)$. For the Civ:V machine, $\Delta$ corresponds to the number of Railroads in the state tiles.}
\def\arraystretch{1.05}
\begin{tabular}{|l|l||c|}
\hline
Civ:V state ($\Delta;\;tape\;read$) \;&\; Command to Workers ($tape;\;state$) & $(10, 3)$-UTM  
 \\ \hline
$\;0$; No Improvement &\; Build a Road and move $R$;       No build            \;&\;  $q_0$ ; $01$, $R q_0$ \\
$\;0$; Road           &\; Remove Improvement and move $L$; Build a Railroad    \;&\;  $q_0$ ; $10$, $L q_1$ \\
$\;0$; Railroad       &\; No Improvement and move $R$;     Build $3$ Railroads \;&\;  $q_0$ ; $bb$, $R q_3$ \\

\hline

$\;1$; No Improvement &\; No Improvement and move $L$;     Build a Railroad \;&\;  $q_1$ ; $00$, $L q_2$ \\
$\;1$; Road           &\; Remove Improvement and move $L$; No build         \;&\;  $q_1$ ; $10$, $L q_1$ \\
$\;1$; Railroad       &\; No Improvement and move $L$;     No build         \;&\;  $q_1$ ; $bb$, $L q_1$ \\

\hline

$\;2$; No Improvement &\; No Improvement and move $L$;   Remove a Railroad   \;&\;$q_2$ ; $00$, $L q_1$\\
$\;2$; Road           &\; Build a Railroad and move $L$; Build $3$ Railroads \;&\; $q_2$ ; $1b$, $L q_5$\\
$\;2$; Railroad       &\; No Improvement and move $R$;   Remove $2$ Railroads\;&\; $q_2$ ; $bb$, $R q_0$\\

\hline

$\;3$; No Improvement &\; Build a Road and move $R$;   Remove $3$ Railroads\;&\; $q_3$ ; $01$, $R q_0$\\
$\;3$; Road           &\; No Improvement and move $R$; Build a Railroad    \;&\; $q_3$ ; $11$, $R q_4$\\
$\;3$; Magrail        &\; Build a Road and move $L$;   No build            \;&\; $q_3$ ; $b1$, $L q_3$\\

\hline

$\;4$; No Improvement &\; Build a Railroad and move $L$; Remove $2$ Railroads \;&\; $q_4$ ; $0b$, $L q_2$ \\
$\;4$; Road           &\; No Improvement and move $R$;   No build             \;&\; $q_4$ ; $11$, $R q_4$ \\
$\;4$; Railroad       &\; No Improvement and move $R$;   No build             \;&\; $q_4$ ; $bb$, $R q_4$ \\

\hline

$\;5$; No Improvement &\; Build a Road and move $L$;   Build a Railroad \;&\; $q_5 0; 1 L q_6$ \\
$\;5$; Road           &\; No Improvement and move $L$; No build         \;&\; $q_5 1; 1 L q_5$ \\
$\;5$; Railroad       &\; No Improvement and move $L$; No build         \;&\; $q_5 b; b L q_5$ \\

\hline

$\;6$; No Improvement &\; No Improvement and move $R$; Build a Railroad   \;&\;  $q_6 0; 0 R q_7$ \\
$\;6$; Road           &\; HALT                                            \;&\;  $\;q_6 1;\;   HALT $ \\
$\;6$; Railroad       &\; No Improvement and move $L$; Build $2$ Railroads\;&\;  $q_6 b; b L q_8$ \\

\hline

$\;7$; No Improvement &\; Build a Road and move $L$;   Remove $2$ Railroads \;&\; $q_7 0;1 L q_5$\\
$\;7$; Road           &\; No Improvement and move $R$; No build             \;&\; $q_7 1;1 R q_7$\\
$\;7$; Railroad       &\; No Improvement and move $R$; No build             \;&\; $q_7 b;b R q_7$\\

\hline

$\;8$; No Improvement &\; Build a Road and move $L$;       Build a Railroad    \;&\; $q_8 0; 1 L q_9$\\
$\;8$; Road           &\; Remove Improvement and move $R$; Build a Railroad    \;&\; $q_8 1; 0 R q_9$\\
$\;8$; Railroad       &\; Remove Improvement and move $L$; Remove $5$ Railroads\;&\; $q_8 b; 0 L q_3$\\

\hline

$\;9$; No Improvement &\; No Improvement and move $R$;     Remove $5$ Railroads\;&\; $q_9 0;0 R q_4$ \\
$\;9$; Road           &\; Remove Improvement and move $R$; No build            \;&\; $q_9 1;0 R q_9$ \\
$\;9$; Railroad       &\; No Improvement and move $R$;     Remove a Railroad   \;&\; $q_9 b;b R q_8$ \\

\hline

\end{tabular}
\end{center}
\end{table}

\subsection{Civ:VI $(48, 2)$-UTM Proof}\label{app:civviapp}

In this section we present a partial proof of \thmref{civviutm}: \tabref{civvi} shows a bijective function between the Civ:VI UTM and Rogozhin's $(24, 2)$-UTM. However, the mechanics around expanding cities are left out, as it is understood that any $R$ or $L$ command would have an extra branching statement to test the population size and build the new Settler, thus making \tabref{civvi} a $(48, 2)$-UTM. Concretely, let $s_t, s'_t$ be states of the machine at some turn $t$, with $n \in s_t$ and $b \in s'_t$ being the only distinction between them. Then $s_t$ is equivalent to executing $s_t$, training a Settler, founding a City, and growing it to four population, and then executing $s_{t+1}$. On the other hand, $s'_t$ would immediately execute $s_{t+1}$.

Each Citizen in a City consumes $2$ Food per turn, so at any point in time the City will need at most $8$ Food per turn. The City itself has a base yield of $2$ ($+2$ since it is on grassland), and both grassland tiles being worked at the same time will provide $4$ Food. This yields a maximum Food per turn of $8$. It is easy to see that this is also the minimum possible amount of Food produced by the City, since floodplains provide $+3$ food. Thus maintaining a City as an operational unit of the tape is feasible. 

Let $L$ be the length of the tape at any given time. Then the number of Cities by the player is given by $L/2$, and the cost of training a new Settler (in Production units) will be $15L + 50$, or $S(L)$ turns. 
Let $C$ be the number of turns required to grow a City to $8$ population. It takes a Settler $3$ turns to found a new City. Then the total overhead for this UTM is given by $C + S(L) + 3$. It follows that this UTM is linearly less efficient than a UTM. 

\begin{table}\label{tab:civvi}
\begin{center}
\caption{Relation between the UTM from \thmref{civviutm} and Rogozhin's $(24, 2)$-UTM, for the states not requiring a new Settler. For the Turing machine, the notation is read as $(state, \;read;\; write,\;move,\;set\;state)$. For the Civ:VI machine, $\Delta$ is the change in Faith from the player's base yield, and the command "Work $n$ Monasteries/Farms" involves reassigning Citizens to Farms or Monasteries as needed.}
\begin{tabular}{|l|l||c|}
\hline
Civ:VI state ($\Delta;\;tape\;read$) \;&\; Command to Workers ($tape;\;state$) & $(24, 2)$-UTM  
 \\ \hline
$\;0$; Is Not Being Worked &\; Is Not Being Worked, move $R$; Work $5$ more Monasteries \;&\;  $q_1 0;0 R q_5$ \\
$\;0$; Is Being Worked     &\; Is Being Worked,     move $R$; Work $2$ more Monasteries \;&\;  $q_1 1;1 R q_2$ \\
$\;1$; Is Not Being Worked &\; Is Being Worked,     move $R$; Work $0$ more Monasteries \;&\;  $q_2 0;1 R q_1$ \\
$\;1$; Is Being Worked     &\; Is Being Worked,     move $L$; Work $2$ more Monasteries \;&\;  $q_2 1;1 L q_3$ \\
$\;2$; Is Not Being Worked &\; Is Not Being Worked, move $L$; Work $2$ more Monasteries \;&\;  $q_3 0;0 L q_4$ \\
$\;2$; Is Being Worked     &\; Is Not Being Worked, move $L$; Work $0$ more Monasteries \;&\;  $q_3 1;0 L q_2$ \\

\hline  

$\;3$; Is Not Being Worked &\; Is Being Worked,     move $L$; Work $9$ more Monasteries \;&\;  $q_4 0;1 L q_{12}$ \\
$\;3$; Is Being Worked     &\; Is Not Being Worked, move $L$; Work $6$ more Monasteries \;&\;  $q_4 1;0 L q_9$ \\
$\;4$; Is Not Being Worked &\; Is Being Worked,     move $R$; Work $3$ more Farms       \;&\;  $q_5 0;1 R q_1$ \\
$\;4$; Is Being Worked     &\; Is Not Being Worked, move $L$; Work $2$ more Monasteries \;&\;  $q_5 1;0 L q_6$ \\
$\;5$; Is Not Being Worked &\; Is Not Being Worked, move $L$; Work $2$ more Monasteries \;&\;  $q_6 0;0 L q_7$ \\
$\;5$; Is Being Worked     &\; Is Being Worked,     move $L$; Work $2$ more Monasteries \;&\;  $q_6 1;1 L q_7$ \\

\hline

$\;6$; Is Not Being Worked &\; Is Being Worked,     move $L$; Work $2$  more Monasteries \;&\;  $q_7 0;1 L q_8$ \\
$\;6$; Is Being Worked     &\; Is Not Being Worked, move $L$; Work $0$  more Monasteries \;&\;  $q_7 1;0 L q_6$ \\
$\;7$; Is Not Being Worked &\; Is Not Being Worked, move $L$; Work $0$  more Monasteries \;&\;  $q_8 0;0 L q_7$ \\
$\;7$; Is Being Worked     &\; Is Being Worked,     move $R$; Work $5$  more Farms       \;&\;  $q_8 1;1 R q_2$ \\
$\;8$; Is Not Being Worked &\; Is Not Being Worked, move $R$; Work $11$ more Monasteries \;&\;  $q_9 0;0 R q_{19}$ \\
$\;8$; Is Being Worked     &\; Is Being Worked,     move $L$; Work $4$  more Farms       \;&\;  $q_9 1;1 L q_4$ \\

\hline

$\;9$; Is Not Being Worked &\; Is Being Worked,     move $L$; Work $5$ more Farms       \;&\;  $q_{10} 0;1 L q_4$ \\
$\;9$; Is Being Worked     &\; Is Not Being Worked, move $R$; Work $4$ more Monasteries \;&\;  $q_{10} 1;0 R q_{13}$ \\
$\;10$;Is Not Being Worked &\; Is Not Being Worked, move $L$; Work $6$ more Farms       \;&\;  $q_{11} 0;0 L q_4$ \\
$\;10$;Is Being Worked     &\; HALT                                                     \;&\;  $\;q_{11} 1;\; HALT $ \\
$\;11$;Is Not Being Worked &\; Is Not Being Worked, move $R$; Work $8$ more Monasteries \;&\;  $q_{12} 0;0 R q_{19}$ \\
$\;11$;Is Being Worked     &\; Is Being Worked,     move $L$; Work $3$ more Monasteries \;&\;  $q_{12} 1;1 L q_{14}$ \\

\hline

$\;12$; Is Not Being Worked &\; Is Not Being Worked, move $R$; Work $2$  more Farms       \;&\;  $q_{13} 0;0 R q_{10}$ \\
$\;12$; Is Being Worked     &\; Is Being Worked,     move $R$; Work $12$ more Monasteries \;&\;  $q_{13} 1;1 R q_{24}$ \\
$\;13$; Is Not Being Worked &\; Is Not Being Worked, move $L$; Work $2$  more Monasteries \;&\;  $q_{14} 0;0 L q_{15}$ \\
$\;13$; Is Being Worked     &\; Is Being Worked,     move $L$; Work $2$  more Farms       \;&\;  $q_{14} 1;1 L q_{11}$ \\
$\;14$; Is Not Being Worked &\; Is Not Being Worked, move $R$; Work $2$  more Monasteries \;&\;  $q_{15} 0;0 R q_{16}$ \\
$\;14$; Is Being Worked     &\; Is Being Worked,     move $R$; Work $3$  more Monasteries \;&\;  $q_{15} 1;1 R q_{17}$ \\

\hline

$\;15$; Is Not Being Worked &\; Is Not Being Worked, move $R$; Work $0$ more Monasteries \;&\;  $q_{16} 0;0 R q_{15}$ \\
$\;15$; Is Being Worked     &\; Is Being Worked,     move $R$; Work $5$ more Farms       \;&\;  $q_{16} 1;1 R q_{10}$ \\
$\;16$; Is Not Being Worked &\; Is Not Being Worked, move $R$; Work $0$ more Monasteries \;&\;  $q_{17} 0;0 R q_{16}$ \\
$\;16$; Is Being Worked     &\; Is Being Worked,     move $R$; Work $5$ more Monasteries \;&\;  $q_{17} 1;1 R q_{21}$ \\
$\;17$; Is Not Being Worked &\; Is Not Being Worked, move $R$; Work $2$ more Monasteries \;&\;  $q_{18} 0;0 R q_{19}$ \\
$\;17$; Is Being Worked     &\; Is Being Worked,     move $R$; Work $3$ more Monasteries \;&\;  $q_{18} 1;1 R q_{20}$ \\

\hline

$\;18$; Is Not Being Worked &\; Is Being Worked,     move $L$; Work $15$ more Farms      \;&\;  $q_{19} 0;1 L q_3$ \\
$\;18$; Is Being Worked     &\; Is Being Worked,     move $R$; Work $0$ more Monasteries \;&\;  $q_{19} 1;1 R q_{18}$ \\
$\;19$; Is Not Being Worked &\; Is Being Worked,     move $R$; Work $1$ more Farms       \;&\;  $q_{20} 0;1 R q_{18}$ \\
$\;19$; Is Being Worked     &\; Is Not Being Worked, move $R$; Work $1$ more Farms       \;&\;  $q_{20} 1;0 R q_{18}$ \\
$\;20$; Is Not Being Worked &\; Is Not Being Worked, move $R$; Work $2$ more Monasteries \;&\;  $q_{21} 0;0 R q_{22}$ \\
$\;20$; Is Being Worked     &\; Is Being Worked,     move $R$; Work $3$ more Monasteries \;&\;  $q_{21} 1;1 R q_{23}$ \\

\hline

$\;21$; Is Not Being Worked &\; Is Being Worked,     move $L$; Work $1$  more Monasteries \;&\;  $q_{22} 0;1 L q_{10}$ \\
$\;21$; Is Being Worked     &\; Is Being Worked,     move $R$; Work $0$  more Monasteries \;&\;  $q_{22} 1;1 R q_{21}$ \\
$\;22$; Is Not Being Worked &\; Is Being Worked,     move $R$; Work $1$  more Farms       \;&\;  $q_{23} 0;1 R q_{21}$ \\
$\;22$; Is Being Worked     &\; Is Not Being Worked, move $R$; Work $1$  more Farms       \;&\;  $q_{23} 1;0 R q_{21}$ \\
$\;23$; Is Not Being Worked &\; Is Not Being Worked, move $R$; Work $10$ more Farms       \;&\;  $q_{24} 0;0 R q_{13}$ \\
$\;23$; Is Being Worked     &\; Is Not Being Worked, move $L$; Work $20$ more Farms       \;&\;  $q_{24} 1;0 L q_3$ \\

\hline

\end{tabular}
\end{center}
\end{table}

\begin{remark}
The linear slowdown provided by the cost of every Settler makes this UTM significantly more inefficient than the other UTMs presented in this paper. 
\end{remark}

This construction can be achieved in-game. See \figref{igcivvi} for an example.

\begin{figure}
\centering
\includegraphics[scale=0.7]{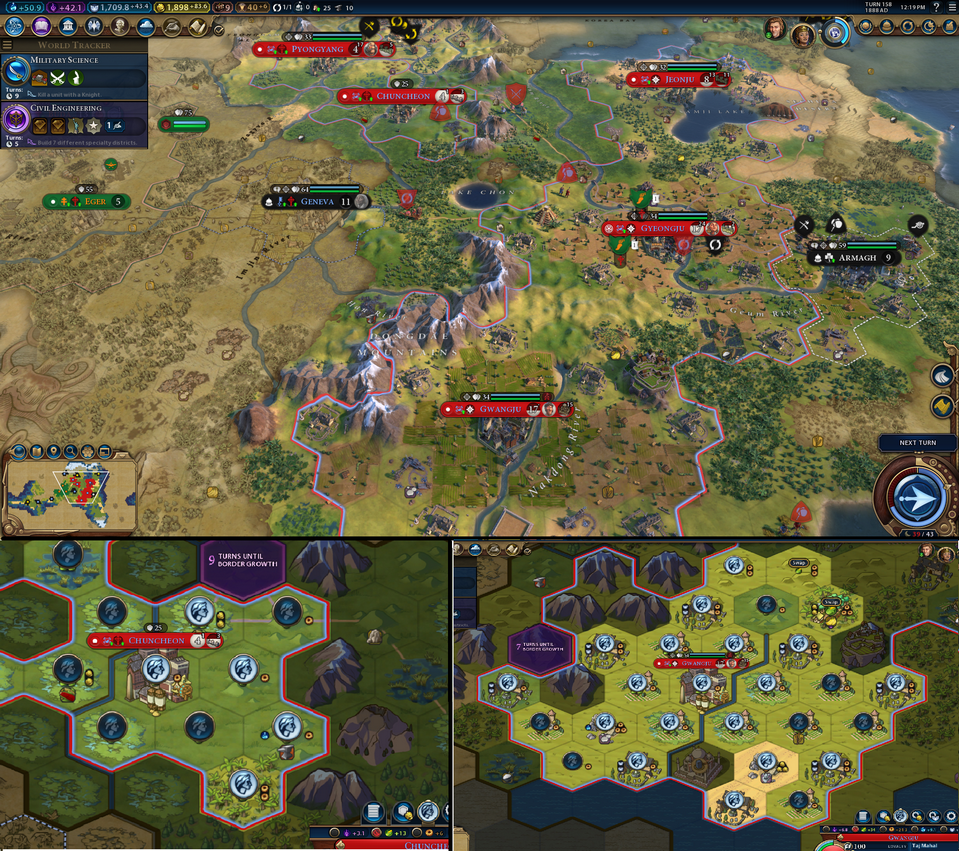}
\caption{A $(48, 2)$-UTM built within Civ:VI. \emph{Top}: the UTM in-game, with three Cities acting as the states and two as the tape. \emph{Bottom right}: one of the Cities holding the state, with $8$ Monasteries and $9$ Farms. \emph{Bottom left:} The tape being read (right hex; Worker not pictured) with a state of "Is Being Worked".}
\label{fig:igcivvi}
\end{figure}

\subsection{An Alternative Civ:VI $(48, 2)$-UTM}\label{app:civviappalt}

In this scenario, we consider two players at war with each other. One player must be playing as Rome, whose unique ability is that Cities founded within $15$ hexes of another City with a trading post, build an automatic road connecting them and a trading post. 

The only change to the construction in \thmref{civviutm} is that instead of tracking symbols by assigning Citizens to hexes, we maintain the road as the tape. The Roman player builds cities and repairs roads with a Worker, and the second player uses a military unit to pillage the roads. Then the set of symbols is $\Gamma = \{$"Has working road", "Has pillaged road"$\}$. Assuming that the players never do anything to affect the computation (e.g., capturing the Workers or Settlers), this construction is equivalent to the one described in the previous sections. The need for a second player is clear given that players cannot pillage improvements within their own borders.

\end{document}